\newif\ifstoc
\stoctrue %last one wins
\stocfalse

\ifstoc
\documentclass[leqno,twocolumn]{article}
\usepackage{ltexpprt}
\else
  \documentclass[letterpaper,11pt]{article}

  \usepackage{typearea}
  \typearea{15}

  \usepackage[compact]{titlesec}
  \usepackage{theorem}
\fi

\newcommand{\stocoption}[2]{{\ifstoc #1 \else #2 \fi}}

\usepackage{typearea}
\typearea{15}

\usepackage{latexsym,graphicx}
\usepackage{amsmath,amssymb,enumerate}
\usepackage{boxedminipage}
\usepackage{xspace}
\usepackage{bm}
\usepackage{ifpdf}
\usepackage{color}
\usepackage[compact]{titlesec}
\usepackage{algorithm}
\usepackage[noend]{algpseudocode}
\usepackage{subfigure}
\usepackage{verbatim}
\usepackage{mathrsfs}
\usepackage{paralist}

\allowdisplaybreaks

\definecolor{Darkblue}{rgb}{0,0,0.4}
\definecolor{Brown}{cmyk}{0,0.81,1.,0.60}
\definecolor{Purple}{cmyk}{0.45,0.86,0,0}
\newcommand{\mydriver}{hypertex}
\ifpdf
 \renewcommand{\mydriver}{pdftex}
\fi
\usepackage[breaklinks,\mydriver]{hyperref}
\hypersetup{colorlinks=true,%pdfborder={1 1 1 [3]},%
            citebordercolor={.6 .6 .6},linkbordercolor={.6 .6 .6},%
citecolor=Darkblue,urlcolor=black,linkcolor=Darkblue,pagecolor=black}
\newcommand{\lref}[2][]{\hyperref[#2]{#1~\ref*{#2}}}

\ifstoc

\else

\newtheorem{theorem}{Theorem}[section]

\newtheorem{lemma}[theorem]{Lemma}
\newtheorem{fact}[theorem]{Fact}

\newtheorem{corollary}[theorem]{Corollary}

\numberwithin{algorithm}{section}

\newenvironment{proof}{

\noindent{\bf Proof:}}
{\hfill$\blacksquare$

}
\fi

\newcommand{\junk}[1]{}
\newcommand{\ignore}[1]{}

\newcommand{\Z}[0]{{\ensuremath{\mathbb{Z}}}}

\def\ceil#1{\lceil #1 \rceil}

\newcommand{\sse}{\subseteq}

\newcommand{\B}{{\mathscr{B}}}
\newcommand{\C}{{\mathscr{C}}}

\newcommand{\G}{{\mathscr{G}}}
\newcommand{\sZ}{{\mathcal{Z}}}
\renewcommand{\S}{{\mathscr{S}}}

\newcommand{\symdif}{\triangle}

\newcommand{\MST}{\ensuremath{\mathsf{mst}\xspace}}
\newcommand{\OPT}{\ensuremath{\mathsf{opt}\xspace}}

\newcommand{\cost}{\ensuremath{\mathsf{cost}\xspace}}

\newcounter{note}[section]

\newcommand{\qedsymb}{\hfill{\rule{2mm}{2mm}}}
\renewenvironment{proof}{\begin{trivlist} \item[\hspace{\labelsep}{\bf
\noindent Proof.\/}] }{\qedsymb\end{trivlist}}%

\newcommand{\initOneLiners}{%
    \setlength{\itemsep}{0pt}
    \setlength{\parsep }{0pt}
    \setlength{\topsep }{0pt}
}
\newenvironment{OneLiners}[1][\ensuremath{\bullet}]
    {\begin{list}
        {#1}
        {\initOneLiners}}
    {\end{list}}

\newcommand{\squishlist}{
 \begin{list}{$\bullet$}
  { \setlength{\itemsep}{0pt}
     \setlength{\parsep}{3pt}
     \setlength{\topsep}{3pt}
     \setlength{\partopsep}{0pt}
     \setlength{\leftmargin}{1.5em}
     \setlength{\labelwidth}{1em}
     \setlength{\labelsep}{0.5em} } }

\newcommand{\squishend}{
  \end{list}  }

\newcommand{\cluster}[2]{\smash{{\C}^{(#1)}_{#2}}}

\newcommand{\vx}[1]{{V({#1})}}
\newcommand{\tr}[1]{T({#1})}

\begin{document}

\title{Online Steiner Tree with Deletions}

\author{
Anupam Gupta\thanks{Computer Science Department, Carnegie Mellon
    University, Pittsburgh, PA 15213, USA, and Microsoft Research SVC,
    Mountain View, CA 94043. Supported in part by
    NSF awards CCF-0964474 and CCF-1016799, and by the
    CMU-MSR Center for Computational Thinking.}
\and Amit Kumar\thanks{Dept. of Computer Science and Engg., IIT Delhi,
  India 110016.}
}
\date{}
\maketitle

\stocoption{}{
\thispagestyle{empty}
\setcounter{page}{0}
%\vspace{-7mm}
}

\begin{abstract}
  In the online Steiner tree problem, the input is a set of vertices
  that appear one-by-one, and we have to maintain a Steiner tree on the
  current set of vertices. The cost of the tree is the total length of
  edges in the tree, and we want this cost to be close to the cost of
  the optimal Steiner tree at all points in time. If we are allowed to
  only add edges, a tight bound of $\Theta(\log n)$ on the
  competitiveness has been known for two decades. Recently it was shown
  that if we can add one new edge and make one edge swap upon every
  vertex arrival, we can still maintain a constant-competitive tree
  online.

  But what if the set of vertices sees both additions and deletions?
  Again, we would like to obtain a low-cost Steiner tree with as few
  edge changes as possible.  The original paper of Imase and Waxman
  (\emph{SIAM J.~Disc.~Math, 4(3):369--384, 1991}) had also considered
  this model, and it gave an algorithm that made at most $O(n^{3/2})$
  edge changes for the first $n$ requests, and maintained a
  constant-competitive tree online. In this paper we improve on these
  results:

  \begin{itemize}
  \item We give an online algorithm that maintains a Steiner tree under
    only deletions: we start off with a set of vertices, and at each
    time one of the vertices is removed from this set---our
    Steiner tree no longer has to span this vertex. We give an algorithm
    that changes only a constant number of edges upon each request, and
    maintains a constant-competitive tree at all times. Our algorithm
    uses the primal-dual framework and a global charging argument to
    carefully make these constant number of changes.
  \item We also give an algorithm that maintains a Steiner tree in the
    fully-dynamic model (where each request either adds or deletes a
    vertex). Our algorithm for this setting makes a constant number of
    changes per request in an amortized sense.
  \end{itemize}
\end{abstract}

%\newpage

\section{Introduction}

The online Steiner tree problem needs little introduction: we have an
underlying metric space,
one-by-one some points in the metric space get
designated as vertices,  and we have to maintain a Steiner tree on the
current set of vertices. The cost of the tree is the total length of
edges in the tree, and we want that at each  time $t$ this cost
stays close to the cost of the optimal Steiner tree on the first $t$
vertices. If we are allowed only to add edges, a tight bound of
$\Theta(\log n)$ on the competitiveness is known~\cite{IW91}. In fact,
this paper considered the ``dynamic'' version of the problem as well,
and asked what would happen if we were allowed to change the Steiner
tree along the way, swapping a small number of previously added edges
for shorter non-tree edges in order to decrease the cost. The rationale
for this problem was natural: much as in dynamic data structures, it
seems natural to ``rewire'' the tree over time, as long as the overhead
in terms of the number of rewirings is not too high, and there is
considerable benefit in terms of cost.

Imase and Waxman~\cite{IW91} showed that one could maintain a constant-competitive
Steiner tree while only making $t^{3/2}$ changes in the first $t$
time-steps; this is a considerable improvement over the naive bound of
$O(t^2)$ obtained by just recomputing the tree at each time-step. More
recently, Megow et al.~\cite{MSVW12} made a significant improvement, showing
that $O(t)$ changes were enough to maintain a constant-competitive tree.
Hence the number of changes required is an \emph{amortized} constant,
improving on the previous $t^{1/2}$-amortized bound.  This was improved
in~\cite{GuGK13} to showing that, in fact, one could make a single edge
swap upon every vertex arrival (in addition to the new edge being added)
--- i.e., a non-amortized bound --- and still achieve
constant-competitiveness.

But what if the set of vertices can undergo both additions and
deletions?  Again, we would like to maintain a low-cost Steiner tree
with as few edge changes as possible. Imase and Waxman~\cite{IW91}
considered this model too; their $t^{3/2}$-change algorithm really is
presented for this more general case where both additions and deletions
take place (the ``fully-dynamic'' case). Very recently~\cite{LOPSZ13}
gave an $O(\log \Delta)$-change algorithm for the fully-dynamic case,
where $\Delta$ is the ratio of the maximum to minimum distance in the
metric. 

\subsection{Our Results}
\label{sec:our-results}

We first consider online algorithms that maintain a Steiner tree under
deletions only. In this model, we start off with a set of vertices, and
at each time one of the vertices is deleted from the this set by an
adversary.  This means our Steiner tree no longer has to span this
vertex --- though we are free to use this deleted vertex in our tree if
we like. In fact we are allowed to maintain any subgraph of the original
graph, as long as all the current vertices lie in a single connected
component of our subgraph. We pay the cost of this subgraph, and want
this cost to be comparable to that of the optimal Steiner tree on the
vertex set. Getting an amortized bound showing a constant number of
changes for this deletions-only case is not difficult (and is
implicit in the~\cite{IW91} paper).  Our first result is an algorithm
that changes only a \emph{non-amortized} constant number of edges upon
each deletion request, and maintains a constant-competitive tree at all
times.

Then we consider the fully-dynamic model, where each request either adds
or deletes a vertex from the current set. For this setting, we give an
algorithm that again maintains a constant-competitive tree, but now the
algorithm makes an \emph{amortized} constant number of changes per
request; i.e., for all $t$, it makes $O(t)$ changes in the first $t$
steps. Getting an algorithm with a non-amortized constant bound for this
setting remains a open problem.

\subsubsection{Our Techniques}
\label{sec:techniques}

%%A word about our techniques.
In order to get an algorithm that makes only a small number of changes
in each step, one natural idea would be that if a vertex is deleted, and
it happens to be a high degree vertex, we just keep it around. This
might be fine because there cannot be too many high degree non-alive
vertices (a vertex is called non-alive if it has been deleted). And if a low degree
 vertex gets deleted, then we remove all
edges incident to it and update the tree (this can only change the tree
by a constant number of edges, proportional to the degree). Of course,
this still requires arguing that the extra cost to retain the
high-degree vertices and their incident edges is small. But there is a
bigger issue --- the above description is inherently incomplete because
deletion of a low degree vertex can cause other high degree non-alive vertices in
the tree to become low degree vertices. Hence, we might not be able to
argue that the tree has only a small number of non-alive low-degree vertices.

Our algorithm for the deletion-only case is essentially a primal-dual
one, and though we state and analyze it purely combinatorially, the
primal-dual viewpoint will be useful to state the intuition. We run a
primal-dual algorithm for Steiner tree at each time-step. However if a
vertex is deleted, we cannot afford to remove it from the primal-dual
process, because this might change the moat-growing and edge-addition in
very unpredictable ways. On the other hand, we cannot afford to keep it
around either: growing a dual moat around it generates ``fake'' duals
--- the dual we generate in this way is not feasible for the cut-packing
dual --- and hence we have to be careful if we want to charge against
this dual. First, for an amortized bound: we use the fact that if a
majority of the dual moats at each point in time are ``real'', we can
globally charge the fake dual to these real duals. If the majority of
these dual moats are ``fake'', then we can afford to make drastic
changes in the Steiner tree -- we can get rid of the vertices in these
fake dual moats, and run the primal-dual process afresh from this point
onwards.  The new edges added can be charged to the deletions
corresponding to the vertices in these fake dual moats.  Thus, we shall
ensure that at any point of time, at least half of the dual moats are
real.

To get a non-amortized bound, we need some more ideas. We will not be
able to make sure that at all times a constant fraction of the dual
moats are real. Instead we will ensure this in a global manner.  We will
show that if such a property does not hold, there will be many fake
moats which can be removed without requiring the tree to change in too
many edges.

For the fully-dynamic case, we use the same greedy algorithm as in the
work of~\cite{IW91}: if there is a tree edge $e$ and a non-tree edge $f$
such that $T - e + f$ is still a spanning tree, we swap $(e,f)$. To
handle deletions, if there is a deleted vertex that has degree one or
two in the current tree, we remove such a low-degree vertex from the
tree; we keep rest of the deleted vertices in our Steiner tree. We show
via a potential function argument (which extends a potential function
from~\cite{GuGK13}) that the number of edge changes performed in $t$
steps is only $O(t)$. This gives the claimed amortized constant-change
constant-competitive algorithm.

\subsection{Other Related Work}
\label{sec:related-work}

The primal-dual method has been used extensively in offline algorithm
design.  Its use for Steiner forest was pioneered by Agrawal, Klein, and
Ravi~\cite{AKR95}, and extended by Goemans and
Williamson~\cite{GW95}. In online algorithm design, the use of
primal-dual techniques is more recent (see, e.g.,~\cite{BN-mono}),
though some early uses of primal-dual in online algorithms are for
online Steiner forest~\cite{AAB96, BC97}. Our use of primal-dual in
online algorithms is of a somewhat different flavor, we show how to
implement a sequence of runs of an offline primal-dual algorithm in an
online fashion. Although this high-level idea also drives the algorithm 
in Gu et~al.~\cite{GuGK13}, the actual details differ in important
ways. Unlike the algorithm in~\cite{GuGK13}, an adversary can now choose
to delete a high degree vertex in the tree, which forces us to keep
such deleted vertices in our tree, and accounting for their cost
requires many of our new technical ideas. Moreover, we now give a global
charging for the primal dual process, which is different from the
techniques in that previous work.

The dynamic Steiner tree was recently also studied by {\L}acki et
al.~\cite{LOPSZ13}, who give a constant-competitive $O(\log
\Delta)$-amortized algorithm for the fully-dynamic case (where $\Delta$
is the ratio of the maximum to minimum distance in the metric). They
also consider the dynamic algorithms problem of maintaining a competitive
Steiner tree, counting not just the number of edges changed, but also
the running time required to maintain such a Steiner tree. For this
problem on planar graphs, they use an combination of online algorithms
and dynamic $(1+\epsilon)$-approximate near-neighbor to give $O(\sqrt{n}
\log^5 n \log \Delta)$-time updates for dynamic Steiner tree. Improving
and extending their results remains an interesting line of research.

Our work is morally related to the work of~\cite{KLS05,KLSZ08} who
consider algorithms for Steiner forest based on a timed analysis: they
give a primal-dual algorithm which also grows fake (infeasible) duals,
but still show that the Steiner forest they create via this dual process
is a $2$-approximation.

\section{Steiner Tree with Deletions: Notation and Definitions}
\label{sec:deletions}

In the case of only deletions, we start off with the set of vertices $V
= \{1, \ldots, n\}$ with a metric $d(\cdot, \cdot)$ on it, and we want
to maintain a Steiner tree connecting them as vertices get deleted one
at a time. At each time-step $t$, one of the vertices gets
deleted. Assuming that our algorithm is oblivious to the names of the
vertices, we assume that the vertex $t$ gets deleted at time-step
$t$. Hence, after the $t^{th}$ time-step we need to ensure that the
vertices in $[t+1 \ldots n]$ are connected. We say that a vertex is
alive at a time $t$ if it has not been deleted in the first $t$
time-steps. So, vertices $[t+1,n]$ are alive at time $t$. For brevity, we
denote the alive vertices after time $t$ by $A_t := [t+1 \ldots n]$, and
the deleted vertices by $D_t := [t] = V \setminus A_t$.

A valid solution at time $t$ is potentially a forest $F_t = (V,E_t)$,
such that there is one tree in this forest that contains all the alive
vertices, i.e., the vertices in $A_t$ are in a connected
component. Hence the initial forest $F_0$ is a spanning tree on $V$. The
cost of a forest $F$ is $\cost(F) := \sum_{e \in F} d(e)$, and our
algorithm is $C$-competitive if at all times $t$ it maintains a forest
$F_t$ such that $\cost(F_t) \leq C \cdot \cost(\MST(A_t))$.  We want to
give a constant competitive algorithm such that the number of changes in
the forest is small. The number of changes at time $t$ is $| F_{t-1}
\symdif F_t |$, and we want the number of changes to be constant either
in an amortized sense (i.e., $\sum_{t \leq T} |F_{t-1} \symdif F_t| =
O(T)$ for all $T \in [n]$), or better still, in a non-amortized sense
(i.e., $|F_{t-1} \symdif F_t| = O(1)$ for all $T \in [n]$).

Observe that if we are giving an amortized bound for the deletions-only
case, it is better to delete all trees in $F_t$ except the one
containing the vertices in $A_t$---this reduces the cost and 
the total number of changes in the forest increase by an additive linear term only. However, this is
not necessarily the case in the non-amortized setting, where the
restrictions on the number of changes means it might make sense to keep
around some components containing only deleted vertices.

\newcommand{\Cnew}{\widehat{\C}}
\newcommand{\thr}{\tau}
\newcommand{\ls}{{\ell^{\star}}}
\newcommand{\nalive}{\#\mathsf{alive}}
\newcommand{\thrt}[2]{\tau^{(#1)}_{#2}}
\newcommand{\bt}[2]{b^{(#1)}_{#2}}
\newcommand{\merge}{{\tt merge}}

\section{Deletions: The Amortized Setting}
\label{sec:amortized-pd}

We first give an amortized algorithm for the deletions-only case, and
then we build on this to give the non-amortized algorithm in
Section~\ref{sec:non-amort-new}. While the algorithm does not explicitly
refer to LPs and duals, there is a clear primal-dual intuition for the
amortized algorithm: when a vertex is deleted, we keep it around as a
``zombie'' node, and grow duals around it. We just ensure that the
number of alive nodes growing duals at any time is at least the number of
zombies growing duals. If this condition is violated, we show how to
remove a set of zombies, and only change a comparable number of tree
edges.

We now describe the process formally. 
For each vertex $v \in V$, we have a bit $b_v \in \{0,1\}$ which says
whether $v$ is not deleted ($b_v = 1$) or deleted ($b_v = 0$). Each
vertex also has a \emph{threshold} $\thr_v$, which is initially set to
$\thr_{\max} := \max_{u,v \in V} \lceil \log_2 d(u,v) \rceil$. We
maintain the invariant that $b_v = 1 \implies \thr_v = \thr_{\max}$.
Moreover, the threshold for a vertex is non-increasing as more vertices
get deleted.

A \emph{cluster} $C=(\vx{C}, \tr{C})$ is a set of vertices $\vx{C}$, along with a spanning tree
$\tr{C}$ joining them. 
Clusters come in three flavors:
\begin{itemize}
\item \textbf{Alive}: Cluster $C$ is \emph{alive} if it contains an undeleted vertex, i.e., at least one $v \in \vx{C}$ has $b_v = 1$.
\item \textbf{Zombie:} $C$ is a \emph{zombie cluster at level $\ell \in
    \Z_{\geq 0}$} if it is not alive, but there is at least one vertex
  $v \in \vx{C}$ which has threshold $\thr_v > \ell$.
\item \textbf{Dead:} $C$ is \emph{a dead cluster at level $\ell$} if it
  is neither alive nor a zombie: i.e., all its vertices have been
  deleted, and all of them have thresholds $\thr_v \leq \ell$.
\end{itemize}
Single vertices are trivially clusters, and hence the same definitions
apply for vertices too. Hence each \emph{deleted} vertex is either a
zombie or it is dead; moreover, a deleted vertex is dead at or above level
$\thr_v$ and a zombie at lower levels. Since being dead or a zombie are
associated with some level $\ell$, we will talk about being $\ell$-dead or
$\ell$-zombie. We also use \emph{$\ell$-non-dead} to mean (alive or
$\ell$-zombie).  

\subsection{A Hierarchical Clustering Algorithm}
\label{sec:formcluster-amort}

A {\em clustering} is a set of clusters such that their vertex sets partition the set
of vertices.  A
hierarchical clustering $\C$ is a collection of clusterings $\C_0, \C_1,
\ldots$, with one clustering $\C_\ell=(\vx{C_\ell}, \tr{C_\ell})$ for each \emph{level} $\ell \in
\Z_{\geq 0}$, such that each cluster $C \in \C_{\ell + 1}$ is the union
of some clusters $C_1, C_2, \ldots, C_p \in \C_\ell$ --- i.e., $\vx{C} =
\cup_{i = 1}^p \vx{C_i}$, and the edges of the tree  $\tr{C}$ are a
super-set of the edges of the trees for each of the $\tr{C_i}$. In other
words, when we combine some $p$ clusters  in $\C_{\ell}$ into $C$, we obtain the tree $\tr{C}$ for $C$
by connecting up the trees for each of the clusters by $p-1$ edges. For
each clustering $\C_i$, there is a forest associated with it, namely $\cup_{C \in \C_i} \tr{C}$, the
union of the trees for each of its clusters: we denote this forest by
$E(\C_i)$.

Our algorithm maintains such a hierarchical clustering.
This algorithm (\textbf{FormCluster}) takes the information
$(b_v, \thr_v)$ for each $v \in V$, and outputs a hierarchical
clustering $\C = \{\C_\ell\}$. Given two clusters $C=(\vx{C}, \tr{C})$ and 
$C'=(\vx{C'}, \tr{C'})$, define  $\merge(C,C')$ as the cluster with 
vertex set $\vx{C} \cup \vx{C'}$ and  spanning tree $\tr{C} \cup \tr{C'} \cup \{e\}$, where
$e$ is the edge between the closest vertices in $V(C)$ and $V(C')$. 
In line 4 of the algorithm, we assume
some consistent ways to break ties and ambiguity.
E.g., when choosing clusters $C, C'$ to merge, we choose the closest
pair $C,C'$, and break ties lexicographically. Also note that $d(C,C')$ refers
to the distance between the closest pair of vertices  $\vx{C}$ and $\vx{C'}$ respectively.

\begin{algorithm}[h!]
  \caption{FormCluster}
  \label{alg:cluster}
  \begin{algorithmic}[1]
    \Require Values $(b_v, \thr_v)$ for all $v \in V$.
    \Statex
    \State Initialize $\C_0$ to be the $n$ clusters with singleton vertex sets $\{1\}, \{2\}, \ldots, \{n\}$.
    \For{$\ell = 0, 1, 2, \ldots$}
    \State $\C_{\ell+1} \gets \C_{\ell}$
    \While{there exist  $C, C' \in \C_{\ell+1}$ which are \\  \hspace*{0.6 cm} $\ell$-non-dead with $d(C,C') \leq 2^{\ell+1}$ }
    \State  $\C_{\ell+1} \gets \C_{\ell+1} \cup \{\merge(C,C')\} \setminus \{ C,C' \} $
    %%\State {\hspace*{0.2 cm} add an edge $e$ between the closest vertices \\ \hspace*{1.5 cm}  of $C, C', i.e., T(C) = T(C) \cup T(C') \cup \{e\}.$}
    \EndWhile
    \EndFor
    \State \textbf{return} the hierarchical clustering $\C = \{\C_\ell\}$
  \end{algorithmic}
\end{algorithm}

There are two special levels associated with this hierarchical clustering: the level
$r(\C)$ is the lowest level such that the clustering $\C_{r(\C)}$
contains a single alive cluster (though there may be other zombie and
dead clusters at level $r(\C)$). The level $s(\C)$ is the lowest level
for which the clustering $\C_{s(\C)}$ contains a single non-dead
cluster, all other clusters are dead at this level. Note that for $\ell
\geq s(\C)$, the clustering $\C_{\ell}$ is the same as $\C_{s(\C)}$, and
hence we can stop the algorithm after this point.

\subsection{The Amortized Algorithm}
\label{sec:amortized-algo}

The main algorithm uses the above clustering procedure. We assume that
all inter-point distances are at least, say, $2$. At time $0$, we start
off with $b_v = 1$ and $\thr_v = \thr_{\max}$ for all $v$.  Let
$\cluster{0}{} \gets \mathbf{FormCluster}(\{(b_v, \thr_v)\}_{v \in V})$,
and $r_0 \gets r(\cluster{0}{})$. Output the tree containing the alive
vertices in the forest $E(\cluster{0}{r_0})$; call this $T_0$.
(Observe: this tree will be a spanning tree on $V$.)

Suppose we have a clustering $\cluster{t-1}{}$, and now vertex $t$ gets
deleted. We first try the lazy thing: just change the bit $b_t$ to $0$,
and run \textbf{FormCluster} to get a hierarchical clustering $\Cnew$.
Observe that for each $\ell$, the clustering $\Cnew_\ell$ is identical to
$\cluster{t-1}{\ell}$, except that possibly one cluster in $\Cnew_\ell$ may
be zombie instead of alive. For level $\ell$, if the number of zombie
clusters in $\Cnew_\ell$ is strictly less than the number of alive
clusters, call the level \emph{good}, else call it \emph{bad}.

\begin{itemize}
\item Case I: Suppose all levels in $\Cnew$ are good, then set
  $\cluster{t}{} \gets \Cnew$ and $T_t \gets T_{t-1}$ and stop.

\item Case II: there are some bad levels in $\Cnew$. Let $\ls$ be
  the lowest bad level.
  Let $\sZ_t$ be
  the set of the zombie clusters in $\Cnew_{\ls}$, and let $Z_t$ be
  the set of vertices in these clusters.
  For each vertex $v \in Z_t$, we set its threshold $\thr_v$ to  $0$.  Hence 
  all nodes in $Z_t$ will be dead at all levels and
  never again take part in cluster formation for future timesteps. Run
  the \textbf{FormCluster} algorithm, now with these new thresholds, to
  get hierarchical clustering $\cluster{t}{}$. Again, $r_t \gets
  r(\cluster{t}{})$, and output the tree containing the alive vertices
  in the forest $E(\cluster{t}{r_t})$; call this tree $T_t$.
\end{itemize}

This completes the description of our algorithm. 
\subsection{The Analysis}
\label{sec:amortized-analysis}

The following facts follows directly from the algorithm above.
%\begin{fact}
%  \label{fct:little-change}
%  Suppose we are in Case~II, and consider $\ell \leq \ls$:
%  \begin{OneLiners}
%  \item If cluster $C \in \Cnew_\ell$ is alive or dead, then $C \in
%    \cluster{t}{\ell}$ with the same status (alive/dead).
%  \item If $C \in \Cnew_\ell$ is a zombie and $C \cap Z_t = \emptyset$,
%    then $C$ is a zombie cluster in $\cluster{t}{\ell}$.
%  \item If $C \in \Cnew_\ell$ is a zombie and and $C \cap Z_t \neq
%    \emptyset$, then $V(C) \sse Z_t$ and all its vertices appear in
%    $\cluster{t}{\ell}$ as singleton dead clusters.
%
%  \end{OneLiners}
%\end{fact}

\begin{fact}
  \label{fct:little-change}
  Suppose we are in Case~II, and consider $\ell <  \ls$:
  \begin{OneLiners}
  \item If a cluster $C \in \Cnew_\ell$ is such that $\vx{C} \cap Z_t = \emptyset $, then 
  $C$ is a cluster in $\cluster{t}{\ell}$ as well with the same status (alive/dead/zombie). 
  Similarly, a cluster $C \in \cluster{t}{\ell}$ with $\vx{C} \cap Z_t = \emptyset $ is also a 
  cluster in $ \Cnew_\ell$ with the same status. 
  \item All vertices in $Z_t$ appear in  $\cluster{t}{\ell}$ as singleton dead clusters. 
  \item If $C \in \Cnew_\ell$ satisfies $\vx{C} \cap Z_t \neq \emptyset $, then $\vx{C} \sse Z_t$.  
    \end{OneLiners}
\end{fact}

\begin{fact}
  \label{fct:only-living}
  In Case~II, for levels $\ell \geq \ls$, all clusters are alive
  or dead; there are no zombies.
\end{fact}

\begin{lemma}
  \label{lem:no-bad}
  There are no bad levels in $\cluster{t}{}$.
\end{lemma}

\begin{proof}
  If we are in Case I and set $\cluster{t}{} = \Cnew$, then we know $\Cnew$
  has no bad levels. Else we are in Case~II, and decrease the thresholds
  of some deleted nodes, and run \textbf{FormCluster} again. For $\ell
  < \ls$, Fact~\ref{fct:little-change} says that each alive
  cluster in $\cluster{t}{\ell}$ corresponds to an alive cluster in
  $\Cnew$, whereas the number of zombie clusters in $\cluster{t}{\ell}$ is
  no more than the number of zombie clusters in $\Cnew_\ell$. Since
  $\ls$ was the lowest numbered bad cluster, level $\ell <
  \ls$ was good in $\Cnew$ and hence is good in $\cluster{t}{}$.
  For $\ell \geq \ls$, the clustering $\cluster{t}{\ell}$
  contains only alive or dead clusters by Fact~\ref{fct:only-living}, so is
  trivially good.
\end{proof}

\begin{fact}
  \label{fct:s-and-c}
  In each clustering $\cluster{t}{\ell}$, all the dead clusters are
  singletons. Moreover, $r(\cluster{t}{}) = s(\cluster{t}{})$.
\end{fact}

\begin{proof}
  When we reduce the threshold for some node, we set it to zero, which
  gives the first statement.  For the second statement, suppose when
  all the alive nodes belong to the same cluster at level
  $r(\cluster{t}{})$, there is another zombie cluster.  Then this level
  would be bad, which would contradict Lemma~\ref{lem:no-bad}.
\end{proof}

\begin{lemma}
  \label{lem:change-amort}
  The number of edges that need to be added or dropped in going from
  $T_{t-1}$ to $T_{t}$ is at most $3|Z_t|$.
\end{lemma}

\begin{proof}
  We must be in Case~II, else $T_t = T_{t-1}$ and there are no edge
  changes. Let $\cluster{t}{\ls}$ have $p$ non-dead clusters,
  which by Fact~\ref{fct:only-living} are all alive. Let $|\sZ_t| = q$ (recall that $\sZ_t$ is the set 
  of zombie clusters  in $\Cnew_{\ls}$) .
  By Fact~\ref{fct:little-change} (and the fact that $\cluster{t-1}{}$
  and $\Cnew$ have the same clusters, modulo some being alive in the former
  and zombie in the latter), the non-dead clusters in $\Cnew_{\ls}$
  and in $\cluster{t-1}{\ls}$ are precisely these $p+q$ clusters.
  By the definition of $\ls$ being a bad level, $q \geq p$.  The
  number of edges that have changed in going from $T_{t-1}$ to $T_t$
  are:
  \begin{OneLiners}
  \item[(a)] Those edges within clusters of $\sZ_t$ are gone; there are
    exactly $|Z_t| - |\sZ_t| = |Z_t| - q$ of these.

  \item[(b)] The $p+q-1$ edges that connect up the clusters in
    $\cluster{t-1}{\ls}$ have potentially been dropped.
  \item[(c)] We add in $p-1$ new edges to connect up the $p$
    alive clusters in $\cluster{t}{\ls}$.
  \end{OneLiners}
  So the total number of edge changes is
    \[ |Z_t| - q + (p+q-1) + p = |Z_t| + 2p - 1 < |Z_t| + 2q.\] Finally,
    note that $q = |\sZ_t| \leq |Z_t|$, so this is less than $3|Z_t|$.
\end{proof}

The above lemma shows that our algorithm makes constant number of changes in the
tree in an amortized sense. Indeed, the set of vertices in $Z_t$ are disjoint for
different values of $t$ -- once a node enters the set $Z_t$, it cannot belong to a 
zombie cluster in subsequent timesteps. 

\begin{fact}
  \label{fct:separation}
  Any two non-dead clusters in $\cluster{t}{\ell}$ are at distance more
  than $2^\ell$ from each other.
\end{fact}

\begin{lemma}
  \label{lem:cost-amort}
  For any $t$, the cost of $T_t$ output by the algorithm is within
  $O(1)$ of the optimal Steiner tree on the alive nodes  $[t+1,n]$.
\end{lemma}

\begin{proof}
  Let $\kappa_{t,\ell}$ be the number of non-dead clusters in
  $\cluster{t}{\ell}$. Observe that $\kappa_{t,\ell} > 1$ for all $\ell <
  r(\cluster{t}{\ell})$ and $\kappa_{t,\ell} = 1$ for all other $\ell$.  Since
  all levels in $\cluster{t}{}$ are good (by Lemma~\ref{lem:no-bad}), we
  know that $\ceil{\kappa_{t,\ell}/2}$ clusters at level $\ell$ are alive
  clusters. And by Fact~\ref{fct:separation}, all these are at distance
  at least $2^\ell$ from each other. A standard dual packing gives a lower
  bound on the cost of the optimal Steiner tree of
  \begin{gather}
    \sum_{\ell \geq 0} (\lceil \kappa_{t,\ell}/2 \rceil - 1) \cdot 2^{\ell-2}
    \geq \frac12 \sum_{\ell = 0}^{r(\cluster{t}{})-1} (\kappa_{t,\ell}/4) \cdot
    2^{\ell-2}. \label{eq:5}
  \end{gather}

  Let $n_{t,\ell}$ be the number of edges added in forming $\cluster{t}{\ell}$
  from $\cluster{t}{\ell-1}$. Hence the cost of the tree $T_t$ is at most
  \[ \sum_{\ell \geq 1} n_{t,\ell} \cdot 2^\ell. \] Moreover, $n_{t,\ell} =
  \kappa_{t,\ell-1} - \kappa_{t,\ell}$, since the number of edges added is
  exactly the reduction in the number of components, so the cost of the
  tree $T_t$ is at most
  \[ \sum_{\ell \geq 1} (\kappa_{t,\ell-1} - \kappa_{t,\ell}) \cdot 2^\ell \leq
  2\kappa_{t,0} + \sum_{\ell = 1}^{r(\cluster{t}{})} \kappa_{t,\ell} 2^\ell. \]
  This is at most a constant times the lower bound~(\ref{eq:5}), which
  proves the result.
\end{proof}

A constant-competitive constant-amortized algorithm for the
deletions-only case can be inferred from the techniques of Imase and
Waxman~\cite{IW91}, so the result is not surprising. But the above
algorithm can be extended to the non-amortized setting, as we show next.

\section{Deletions: The Non-Amortized Setting}
\label{sec:non-amort-new}
We now describe our algorithm in the non-amortized setting. 
Our algorithm is a direct extension of the one above, so
let us think about why we get a large number of changes. This could happen 
because of two reasons. Firstly, if there were
some deletion that caused a large zombie cluster to be marked dead, we
would remove all the edges within the tree connecting this cluster and
hence make a large number of changes. The main observation is that since
we could pay for all the edges within this tree in the previous step, we
should be also able to pay for most of them at the next step, and it
should suffice to remove a constant number of edges. To do this, we will not just set the thresholds
to $\tau_{\max}$ or $0$, but will slowly lower them.

Secondly, we happened to mark a 
small zombie cluster dead, but it was being used to connect many other clusters. 
We get around this problem by marking only those zombie clusters dead which would be 
used for connecting a small number of clusters in subsequent steps -- we show that 
it is always possible to find such zombie clusters.

\subsection{A Modified Cluster-Formation Algorithm}
\label{sec:formcluster-new}

The first change from the previous algorithm is that we want the tree
$T_t$ to be similar to $T_{t-1}$. So we explicitly ensure this by being
as similar to an ``old'' clustering $\C'$ given as input; the algorithm 
is otherwise very similar to Algorithm \textbf{FormCluster}, and we assume
the reader is familiar with that section. Again, let $E(\C_\ell)$ be the
edges contained in the forest corresponding to a clustering $\C_\ell$.

\begin{algorithm}[h!]
  \caption{FormClusterNew}
  \label{alg:cluster-lip}
  \begin{algorithmic}[1]
    \Require Values $(b_v, \thr_v)$ for all $v \in V$, old hierarchical clustering $\C'$.
    \Ensure A hierarchical clustering $\C = \{\C_\ell\}$.
    \Statex
    \State  Initialize $\C_0$ to be the $n$ clusters with singleton vertex sets $\{1\}, \{2\}, \ldots, \{n\}$.
    \For{$j = 0, 1, 2, \ldots$}
    \State $\C_{j+1} \gets \C_{j}$
    \While{there exists edge $e \in E(\C'_{j+1})$ between \\
     \hspace*{0.7 cm}  $j$-non-dead clusters $C, C' \in \C_{j+1}$}
    \State \hspace*{1 mm} Define a new cluster $C''$ with 
     \State \hspace*{3 mm} $\vx{C''} = \vx{C} \cup \vx{C'}$ 
    \State \hspace*{3 mm} and $\tr{C''} = \tr{C} \cup \tr{C'} \cup \{e\}$.
    \State \hspace*{1 mm} $\C_{j+1} \gets \C_{j+1} \cup \{ C''  \} \setminus \{C,C'\}$ \\
    \Comment{\emph{the cluster $C''$ is $j$-non-dead}}
    %%\State \hspace*{1 mm}  add edge $e$ between $C, C'$
    \EndWhile
    \While{there exist $j$-non-dead clusters $C, C' \in \C_j$ \\ \hspace*{0.7 cm}  with $d(C,C') \leq 2^j$ }
    \State  \hspace*{1 mm} $\C_{j+1} \gets \C_{j+1} \cup \{ \merge{C,C'} \} \setminus \{C,C'\}$ \\
    \Comment{\emph{again, $\merge(C,C')$ is $j$-non-dead}}
    %%\State  \hspace*{1 mm}  add an edge between the closest vertices \\ \hspace*{1.5 cm} of $C, C'$
    \EndWhile
    \EndFor
    \State \textbf{return} the new hierarchical clustering $\C = \{\C_\ell\}$
  \end{algorithmic}
\end{algorithm}

Again there are two special levels: level $r(\C)$ is the lowest level
such that $\C_{r(\C)}$ contains a single alive cluster, and level
$s(\C)$ is the lowest level where $\C_{s(\C)}$ contains a single
non-dead cluster, all other clusters are dead at this level.

\subsection{The Non-Amortized Algorithm}
\label{sec:non-amortized-algo}

Again, assume that all inter-point distances are at least $2$. At time
$0$, start off with $b_v = 1$ and $\thr_v = \thr_{\max}$ for all $v$.
Let $\cluster{0}{} \gets \mathbf{FormCluster}(\{(b_v, \thr_v)\}_{v \in
  V})$, and $r_0 \gets r(\cluster{0}{})$. Output the tree containing the
alive vertices in the forest $E(\cluster{0}{r_0})$; call this $T_0$.

Consider the clustering $\cluster{t-1}{}$ corresponding to thresholds
$\thrt{t-1}{}$ and bits $\bt{t-1}{}$. Now vertex $t$ is deleted, so set
$\bt{t}{t} = 0$, and $\bt{t}{j} = \bt{t-1}{j}$ otherwise, and run
\textbf{FormClusterNew}$(\thrt{t-1}{}, \bt{t}{})$ to get a new
hierarchical clustering $\Cnew$. As before, $\cluster{t-1}{}$ and $\Cnew$ will
be identical at all levels, except for perhaps one cluster at each level
being alive in the former and zombie in the latter. The algorithm \textbf{FormClusterNew} is
described above. 
% Again, note that we are using a fixed, though arbitrary, rule to find the 
% clusters $C,C'$. Also when we replace $C$ and $C$ by $C \cup C'$, the spanning tree for 
% $C \cup C'$ is obtained by taking the union of the two corresponding spanning trees of $C$ and $C'$ and 
% a new edge (added in Step 8 or Step 13).

The definition of a level being good is slightly different now. We first develop some more
notation. Denote the edges added at level $\ell$ of any
clustering $\C$ as $E_{\ell}(\C)$; i.e., $E_{\ell}(\C) := E(\C_{\ell})
\setminus E(\C_{\ell - 1})$. Let $m_{\ell}(\C) := |E_{\ell}(\C)|$ be the
cardinality of this set, and $m_{>\ell}(\C) := \sum_{j > \ell} m_j(\C)$
denote the edges added at levels (strictly) above level $\ell$, up to
and including level $s(\C)$.  (Note this includes edges above level
$r(\C)$; of course no edges are added above level $s(\C)$.) Finally, let
$\nalive(\C_{\ell})$ denote the number of alive clusters in $\C_\ell$.

A level $\ell$ of a hierarchical clustering $\C$ is \emph{good} if
\begin{gather}
  m_{> \ell}(\C) \leq 3\, \nalive(\C_\ell); \label{eq:6}
\end{gather}
i.e., if the number of edges
added above level $\ell$ is at most three times the number of alive
components at level $\ell$. Note that if the number of  alive components were more
than a third of the number of non-dead components, the level would be
good. But since we will now use the thresholds in a more nuanced way, a
level may be good even if almost all the non-dead components are zombies.

One final definition: given a hierarchical clustering $\C$ and a level
$\ell$, the \emph{level-$\ell$ skeleton} $\G_\ell(\C)$ is an undirected graph defined as
follows. The vertex set is the set of clusters of $\C_\ell$. There is an
edge in $\G_\ell(\C)$ connecting two clusters $C, C' \in \C_\ell$
precisely when there is an edge between $C, C'$ in the set $\cup_{\ell'
  > \ell} E_{\ell'}(\C)$. In other words, if we were to take the
clustering $\C_{s(\C)}$ and collapse each cluster $C \in \C_\ell$ into a
single node, we would get $\G_\ell(\C)$. By our construction, the
skeleton is always a forest, and contains $m_{> \ell}(\C)$ edges. The
\emph{degree} of a cluster $C \in \C_{\ell}$ is the degree of the
corresponding node in $\G_{\ell}(\C)$.

Now we are ready to state the algorithm. Recall that we constructed $\Cnew
\gets \mathbf{FormClusterNew}(\thrt{t-1}{}, \bt{t}{})$. Again, there are
two cases:
\begin{itemize}
\item Case I: Suppose all levels in $\Cnew$ are good, or for every bad
  level $\ell$ we have $m_{>\ell}(\Cnew) < 36$. Then set $\cluster{t}{}
  \gets \Cnew$.

\item Case II: There exists a bad level $\ell$ such that $m_{>\ell}(\Cnew)
  \geq 36$. Let $\ls$ be the \emph{highest} such level, and consider the
  skeleton $\G_{\ls}(\Cnew)$. Choose a set $\sZ_t$ of $6$ zombie clusters
  from $\Cnew_{\ls}$ that have degree $1$ or $2$ in
  $\G_\ls(\Cnew)$.\footnote{Such a set exists by Fact~\ref{fact:tree}:
    since $\G_\ls(\Cnew)$ is a forest of at least $36$ edges, if we choose
    $A$ to be the set of alive clusters, Fact~\ref{fact:tree} implies
    there must be $6$ non-alive clusters of degree one or two. Since
    these have non-zero degree in $\Cnew_{\ls}$, they cannot be dead at
    this level and must be zombies.} Let $Z_t$ be the vertices in these
  clusters. For each $v \in Z_t$, set $\thrt{t}{v} \gets
  \min(\thrt{t-1}{v}, \ls)$. Observe that these nodes  now form singleton dead
clusters  at
  level $\ls$ according to the new thresholds $\thrt{t}{}$, whereas 
at least one of them in each cluster must have had threshold above $\ls$,
 i.e., $\forall C \in \sZ_t,
  \exists v \in \vx{C}: \thrt{t-1}{v} > \ls = \thrt{t}{v}$. This is true because these clusters 
  are non-dead.  So we're making progress in terms of strictly decreasing the threshold for some
  vertices.  Now set $\cluster{t}{} \gets
  \mathbf{FormClusterNew}(\thrt{t}{}, \bt{t}{})$.
\end{itemize}

In either case, let $r_t \gets r(\cluster{t}{})$ be the lowest level
with a single alive cluster, and return the \emph{forest} corresponding
to this level --- i.e., $F_t \gets E(\cluster{t}{r_t})$.

For future convenience, define $s_t \gets s(\cluster{t}{})$ to be the
lowest level with a single non-dead cluster, so all the edges in
$\cluster{t}{}$ belong to $E(\cluster{t}{s_t})$. In Case~I, $r_t$ may be
much smaller than $r_{t-1}$ but $s_t = s_{t-1}$. Moreover, let $F'_t$
be the edges added at all levels of the algorithm --- so $F_t'
\setminus F_t$ are the edges added in levels $r_t+1, \ldots, s_t$.

Before we begin the analysis, notice the differences between this
algorithm and the previous amortized one: previously, badness meant the
number of zombies was more than the alive clusters at that level, now
badness means the number of edges being added above the level is much
more than the number of alive clusters in that level. Previously, we
chose all the zombie clusters at the \emph{lowest} bad level and made 
their nodes dead right at level $0$. Now we choose the \emph{highest} bad
level $\ls$ and carefully choose some six zombie clusters, and make
their nodes dead only at level $\ls$ --- this will ensure that only a small
number of edges will change between timesteps $t$ and $t+1$. 

\subsection{The Analysis}
\label{sec:nonamort-analysis}

At a high level, the analysis will proceed analogously to
Section~\ref{sec:amortized-analysis}, but the details are more
interesting. If we are in Case I, things are simple. Indeed, for each
$\ell$, the clusterings $\cluster{t}{\ell}$ and $\cluster{t-1}{\ell}$
have the same clusters, apart from potentially one cluster being zombie
in the former and alive in the latter. Since bad levels, if any, satisfy
$m_{> \ell}(\cluster{t}{}) < 36$, we get that for every level $\ell$, we
have $m_{> \ell}(\cluster{t}{}) < 3\nalive(\cluster{t}{\ell}) + 36$ if
$\cluster{t}{}$ was produced from Case~I.

The bulk of the work will be to show an analogous inequality for
Case~II. Here, we first show that the structure of the clusterings at
times $t-1$ and $t$ differ only in a controlled fashion. In fact, we
show that a large number of ``safe'' edges will be common to $F_t$ and
$F_{t-1}$. This will allow us to show that all levels in the clustering
we output may not be good, we still have $m_{> \ell}(\cluster{t}{}) <
3\nalive(\cluster{t}{\ell}) + O(1)$ in this case.

And why is such a bound useful? If there is a single living cluster at
some level, only a constant number of edges are added above this level,
and dropping them will change only a constant number of edges. On the
other hand, if more than one cluster is alive, we are able to pay for
the edges added at this level. Putting all this together will ensure
that $|F_t \symdif F_{t-1}|$ is bounded, and $F_t$ is constant
competitive.

\subsubsection{The Structure of Clusters}

\begin{lemma}
  \label{lem:refinement}
  The clusters in $\cluster{t}{\ell}$ are a refinement of the clusters
  in $\cluster{t-1}{\ell}$:  for each  cluster $C \in
  \cluster{t-1}{\ell}$, $\vx{C}$ is the union of  vertex sets $\vx{C_1}, \ldots, 
\vx{C_p}$ corresponding to  some   clusters $C_1,
  C_2, \ldots, C_p \in \cluster{t}{\ell}$. 

  Suppose we are in  Case~II. Then each of the clusters in $\sZ_t$ belong to
  $\cluster{t}{\ell}$ for all $\ell \geq \ls$, and they are dead
  clusters at these levels.
  Moreover, if cluster $C \in
  \cluster{t-1}{\ell}$ is non-dead, then each of the corresponding $C_i
  \in \cluster{t}{\ell}$ are either in $\sZ_t$ or are non-dead.
\end{lemma}

\begin{proof}
  The lemma is immediate in Case~I, so assume we are in Case~II. We
  prove the lemma by induction on $\ell$. If $\ell \leq \ls$, then since
  we only reduced the thresholds of some nodes down to $\ls$ and changed
  $b_t = 0$, clusters that were non-dead according to
  $(\thrt{t-1}{\ell}, \bt{t-1}{\ell})$ are also non-dead according to
  $(\thrt{t}{\ell}, \bt{t}{\ell})$, and there is no change in the
  actions of \textbf{FormClusterNew} for these ``low'' levels. For such
  levels $\ell \leq \ls$, there is a bijection between clusters in
  $\cluster{t-1}{\ell}$ and $\cluster{t}{\ell}$.

  %% Further if $C$ were alive in
  %% $\cluster{t}{l^\star}$, it would be alive in
  %% $\cluster{t+1}{l^\star}$, unless
  %% the only alive vertex in it was $t+1$. But then, it
  %% would be declared a zombie cluster in
  %% $\cluster{t+1}{l^\star}$.
  Now consider level $\ell > \ls$, and assume the first (refinement)
  claim is true for level $\ell - 1$. Now suppose we add an edge between the spanning 
trees $\tr{C_1}$ and $\tr{C_2}$ of 
  two clusters $C_1, C_2$ respectively  in  $\cluster{t}{\ell-1}$ (to form a cluster at
  level $\ell$), and hence $d(C_1, C_2) \leq 2^\ell$. By the inductive
  hypothesis, $\vx{C_1}$ is contained in  $\vx{C_1'}$, $C_1' \in
  \cluster{t-1}{\ell-1}$ and $\vx{C_2}$ is contained in  $\vx{C_2'}, C_2' \in
  \cluster{t-1}{\ell-1}$, and so $d(C_1',C_2') \leq 2^{\ell}$. Hence, 
  $C_1'$ and $C_2'$ will become part of the same cluster in 
  %%$\vx{C_1'}, \vx{C_2'}$ will also be subsets of the same cluster in
  $\cluster{t-1}{\ell}$. This proves the refinement claim for level
  $\ell$.

  For the second statement, the clusters in $\sZ_t \sse
  \cluster{t-1}{\ls}$ all exist in $\cluster{t}{\ls}$ (by the above
  claim about a bijection for levels $\ell \leq \ls$), and they are all
  $\ls$-dead in $\cluster{t}{\ls}$ (by construction of $\thrt{t}{}$), so
  they will remain in all subsequent levels $\ell \geq \ls$.

  For the last statement, for $\ell \geq \ls$, suppose $C \in
  \cluster{t-1}{\ell}$ corresponds to $C_1, C_2, \ldots, C_p \in
  \cluster{t}{\ell}$, and $C_i$ is dead.  The only possibilities for
  $C_i$ are (a) it was already dead in $\cluster{t-1}{}$, in which case $C
  = C_i$ and $C$ will be dead as well, which is a contradiction, or
  (b)~$C_i$ is one of the clusters in $\sZ_t$.
\end{proof}

\subsubsection{Safe Edges}

The results of this section are interesting only when we are in Case~II,
and $\ls$ is defined.  Consider the skeleton $\G_\ls(\cluster{t-1}{})$;
recall that the clusters in $\sZ_t$ correspond to degree-one or degree-two nodes in
this graph. We define a set of \emph{boundary} clusters $\B_t$ to be
those clusters in $\G_\ls(\cluster{t-1}{})$ that are not in $\sZ_t$ but
have at least one cluster in $\sZ_t$ as a neighbor. Hence $|\B_t| \leq
2|\sZ_t|$.
% Since the clusters in $\sZ_t$ also exist
% in $G_\ls(\cluster{t}{})$ (by the last claim in
% Lemma~\ref{lem:refinement}), we can define $\B_

An edge $e$ in $E(\cluster{t-1}{s_{t-1}})$ (which is the set of all
edges added in the hierarchical clustering $\cluster{t-1}{}$, regardless
of whether it was part of $F_{t-1}$ or not) is called \emph{safe} if
either (a)~$e$ belongs to $E(\cluster{t-1}{\ls})$, i.e., it was added at
level $\ell \leq \ls$, or (b)~at least one endpoint of $e$ belongs to a
cluster not in $\sZ_t \cup \B_t$. In other words, an edge is unsafe if
and only if it belongs to $\G_\ls(\cluster{t-1}{})$ and both endpoints
fall in clusters in $\sZ_t \cup \B_t$.

\begin{fact}
  \label{fct:most-safe}
  At most $3|\sZ_t|-1$ edges are unsafe.
\end{fact}
\begin{proof}
  The edges in $\G_\ls(\cluster{t-1}{})$ form a forest, and an unsafe
  edge is a subset of these edges that has both endpoints in $\sZ_t \cup
  \B_t$. So there are at most $|\sZ_t \cup \B_t| - 1$ unsafe edges.
  Moreover, each cluster in $\sZ_t$ has at most two neighbors, so
  $|\B_t| \leq 2|\sZ_t|$, which proves the claim.
\end{proof}

\begin{lemma}
  \label{lem:safe-stable}
  If $e \in E_\ell(\cluster{t-1}{})$ is a safe edge, then $e \in
  E_\ell(\cluster{t}{})$. In other words, every safe edge added at level
  $\ell$ at time $t-1$ is added at level $\ell$ at time $t$.
\end{lemma}

\begin{proof}
  For $\ell \leq \ls$, this follows because the algorithms at time $t-1$
  and time $t$ behave the same until level $\ls$: every edge is safe,
  and is added at the same time.  For $\ell > \ls$, consider a safe edge
  $e = (x,y) \in E_\ell(\cluster{t-1}{})$ going between clusters
  $C_{t-1}, C_{t-1}' \in \cluster{t-1}{\ell}$. Let $C_{t}, C_{t}'$ be
  the clusters in $\cluster{t}{\ell}$ containing $x$ and $y$
  respectively.  Lemma~\ref{lem:refinement} implies that $C_{t}
  \subseteq C_{t-1}$ and $C_{t}' \subseteq C_{t-1}'$.

  First, observe that $x, y \notin Z_t$ (where $Z_t$ is the set of
  vertices lying in the clusters of $\sZ_t$). Indeed, if $x \in Z_t$,
  then the cluster in $\cluster{t-1}{\ls}$ containing $x$ at level $\ls$
  would belong to $\sZ_t$, and then $y \in \sZ_t \cup \B_t$, since
  $\B_t$ contains all the neighboring clusters of $\sZ_t$ in
  $\cluster{t-1}{}$. This contradicts $(x,y)$ being safe.

  We now claim that both $C_{t}$ and $C_{t}'$ are non-dead in
  $\cluster{t}{\ell}$. 
  Suppose $C_{t}$ was dead. Note that  $C_{t}
  \not\in \sZ_t$ because $x \notin Z_t$. Therefore, the second part of Lemma~\ref{lem:refinement} implies
  that $C_{t-1}$ would be dead in $\cluster{t-1}{\ell}$. But then the
  edge $(x,y)$ would not be added, a contradiction. A similar argument
  shows that $C_{t}'$ is not dead. Moreover, since the clustering at time $t$ is a
  refinement of that at time $t-1$ (again by
  Lemma~\ref{lem:refinement}), adding the edge $e$ to $E_\ell(\cluster{t-1}{})$ will not create a cycle.
  Hence, we will add $e$ to $E_\ell(\cluster{t}{})$. 
\end{proof}

To summarize, there are very few edges that are unsafe
(Fact~\ref{fct:most-safe}), and safe edges are added at the same level at
timestep $t$ as at timestep $t-1$. This will be useful to show that the
edge set in consecutive steps remains pretty similar.

\subsubsection{Bounding the Changes}

Let us define some syntactic sugar. Let the number of alive clusters in
$\cluster{t}{\ell}$ be denoted $a_{t,\ell}$ instead of
$\nalive(\cluster{t}{\ell})$. Let the number of edges added at levels
above $\ell$ in $\cluster{t}{}$ be denoted by $m_{t, >\ell}$ instead of
$m_{>\ell}(\cluster{t}{})$.

\begin{lemma}
  \label{lem:alive}
  For all levels $\ell$, $a_{t,\ell} \geq
  a_{t-1,\ell} - 1$.
\end{lemma}
\begin{proof}
  The clustering $\cluster{t}{\ell}$ is a refinement of the clustering
  $\cluster{t-1}{\ell}$, so each alive cluster in $\cluster{t-1}{\ell}$
  gives rise to at least one alive cluster in $\cluster{t}{\ell}$ ---
  except for the cluster containing vertex $t$, which might become a
  zombie at time $t$, and accounts for the subtraction of one.
\end{proof}

\begin{lemma}
  \label{lem:forest}
  The difference in the total number of edges added at timesteps $t-1$ and
  $t$ is
  \begin{gather}
    | F_{t-1}' | - | F_t' | \geq
    |\sZ_t|/2 = 3.
    \label{eq:7}
  \end{gather}
  Moreover:
  \begin{alignat}{3}
    m_{t, >\ell} &\leq m_{t-1, >\ell} -
    |\sZ_t|/2 &  &\leq m_{t-1, >\ell} - 3 &\quad &\forall \ell \leq \ls \label{eq:1} \\
    m_{t, >\ell} &\leq  m_{t-1, >\ell} + 3|\sZ_t|
    & &\leq m_{t-1, >\ell} + 18
    &\quad &\forall \ell > \ls \label{eq:2}
  \end{alignat}
\end{lemma}
\begin{proof}
  Pick level $\ell_M = \max(s_{t-1}, s_t)$. We claim the difference
  in the number of clusters at level $\ell_M$ is
  \begin{gather}
    |\cluster{t}{\ell_M}| - |\cluster{t-1}{\ell_M}| \geq |\sZ_t|/2
    .\label{eq:4}
  \end{gather}
  To see this, observe that the vertex set in each cluster in $\cluster{t-1}{\ell_M}$ is
  union of the vertex sets of some clusters in $\cluster{t}{\ell_M}$ by
  Lemma~\ref{lem:refinement}, so the difference above is definitely
  non-negative. Moreover, each of the clusters in $\sZ_t$ forms an
  isolated cluster in $\cluster{t}{\ell_M}$, but it used to have
  positive degree in $\cluster{t-1}{\ell_M}$. The extreme case is when
  these clusters induce a matching, but that still increases the number
  of clusters by $|\sZ_t|/2$. This proves~(\ref{eq:4}).

  For any level $\ell$, the quantity $m_{t-1,>\ell} = m_{>\ell}(\cluster{t-1}{\ell})$ is
  the number of edges added above level $\ell$, which is equal to the
  reduction in the number of clusters above this level. Hence
  $m_{t-1,>\ell} = |\cluster{t-1}{\ell}| -
  |\cluster{t-1}{\ell_M}|$. Similarly $m_{t,>\ell} =
  |\cluster{t}{\ell}|-|\cluster{t}{\ell_M}|$.  Since
  $|\cluster{t-1}{\ell}|=|\cluster{t}{\ell}|$ for $\ell \leq \ls$,
  we have
  \stocoption{
    \begin{multline*} m_{t-1, >\ell} - m_{t, >\ell} \\ =
      (|\cluster{t-1}{\ell}| - |\cluster{t-1}{\ell_M}|) -
      (|\cluster{t}{\ell}| - |\cluster{t}{\ell_M}|) \geq |\sZ_t|/2, 
    \end{multline*}
}{
    \[ m_{t-1, >\ell} - m_{t, >\ell} = (|\cluster{t-1}{\ell}| -
    |\cluster{t-1}{\ell_M}|) - (|\cluster{t}{\ell}| -
    |\cluster{t}{\ell_M}|) \geq |\sZ_t|/2, \] } the last
  from~(\ref{eq:4}). This proves~(\ref{eq:1}). Also, $|F_{t-1}'| -
  |F_t'| = m_{t-1,\geq0} - m_{t,\geq 0} = m_{t-1,>\ls} - m_{t,>\ls} \geq
  |\sZ_t|/2$, and so~(\ref{eq:7}) also follows.

  For a level $\ell \geq \ls$, all the safe edges at $\ell$ and lower
  levels in time $t-1$ get added at the corresponding level in  time $t$ as well (Lemma~\ref{lem:safe-stable}). To maximize the
  difference, it can only be the case that all the unsafe edges (of
  which there are at most $3|\sZ_t|$) might not have been added yet.
  This proves~(\ref{eq:2}). Plugging in $|\sZ_t| = 6$ gives the
  numerical values.
\end{proof}

\subsubsection{The Key Invariant}

We now prove the key invariant. In the amortized case, we could prove
that for each hierarchical clustering $\cluster{t}{}$, all levels were
good. In the non-amortized case, this will not be true. However, we will
show a slightly weaker invariant. Recall the notion of
goodness~(\ref{eq:6}): for any clustering $\C$, level $\ell$ is good if
$ m_{> \ell}(\C) \leq 3 \, \nalive(\C_\ell)$. Using our shorthand, a
good level $\ell$ for timestep $t$ means $m_{t, >\ell} \leq
3\,a_{t,\ell}$. What about bad levels?
\begin{lemma}[Invariant]
  \label{lem:invariant}
  For all timesteps $t$, if level $\ell$ is bad for
  $\cluster{t}{}$, and $m_{t, >\ell} \geq 36$.  Then
  \begin{eqnarray}
    \label{eq:invariant}
    m_{t, >\ell} \leq 3 a_{t,\ell} + 54.
  \end{eqnarray}
\end{lemma}

\begin{proof}
  We prove this by induction on $t$. Initially, at time $t=0$, all
  vertices are alive. For any level $\ell$, the number of edges added
  above that level can be at most the number of components at that
  level. Thus $m_{0,> \ell} \leq a_{0,\ell} - 1$. This means all levels
  are good, and the invariant is vacuously true.

  Suppose~(\ref{eq:invariant}) holds true at some time $t-1$ for all bad
  levels $\ell$. We need to show that~(\ref{eq:invariant}) holds at time
  $t$ for all bad levels $\ell$ as well. If we were in Case~I, then we
  know that $ m_{t, >\ell} < 3 a_{t,\ell} +
  36$ (since either all levels of $\Cnew$ were good, or they had $ m_{>
    \ell}(\cluster{t}{}) < 36$).

  Hence we need to consider when we get to time $t$ using Case~II.  Let
  $\ls$ be as defined by the algorithm --- the highest bad
  level $\ell$ in the intermediate hierarchical clustering $\Cnew$ with
  $m_{> \ell}(\Cnew) = m_{t-1, >\ell} \geq 36$.

  Now take $\cluster{t}{}$, and first consider a bad level for some
  $\ell \leq \ls$. There are several cases.
  \begin{itemize}
  \item Suppose $\ell$ was a good level in $\cluster{t-1}{}$: by
    definition of goodness, $m_{t-1, >\ell} \leq 3
    a_{t-1,\ell}$. 
    Therefore,
    \stocoption{
    \begin{multline*}
      m_{t, >\ell}
      \stackrel{\small{(\ref{eq:1})}}{\leq}
      m_{t-1, >\ell} - 3 
      \stackrel{\small{\text{goodness}}}{\leq}
      3 a_{t-1,\ell} - 3 \\
      \stackrel{\small{\text{Lemma}~\ref{lem:alive}}}{\leq}
      3( a_{t,\ell} + 1) - 3 <  3a_{t,\ell} + 54 .
    \end{multline*}
  }{
    \begin{gather*}
      m_{t, >\ell}
      \stackrel{\small{(\ref{eq:1})}}{\leq}
      m_{t-1, >\ell} - 3
      \stackrel{\small{\text{goodness}}}{\leq}
      3 a_{t-1,\ell} - 3
      \stackrel{\small{\text{Lemma}~\ref{lem:alive}}}{\leq}
      3( a_{t,\ell} + 1) - 3 <  3a_{t,\ell} + 54 .
    \end{gather*}
  }
  \item Suppose $\ell$ was a bad level in $\cluster{t-1}{}$, but
    $m_{t-1,>\ell} < 36$: in this case,
    $$m_{t,>\ell}
      \stackrel{\small{(\ref{eq:1})}}{\leq}
      m_{t-1,>\ell} - 3 < 36 - 3, $$ and so the invariant
    holds trivially.
  \item Finally, suppose $\ell$ was a bad level in time $t-1$, and
    $m_{t-1,>\ell} \geq 36$: we can now apply the invariant at time $t-1$ to
    this level $\ell$. So, we get
    \stocoption{
      \begin{multline*}
    m_{t,>\ell}
      \stackrel{\small{(\ref{eq:1})}}{\leq}
      m_{t-1,>\ell} - 3
      \stackrel{\small{\text{invariant}}}{\leq}
      3 a_{t-1,\ell} + 54 - 3 \\
      \stackrel{\small{\text{Lemma}~\ref{lem:alive}}}{\leq}
      3 (a_{t,\ell}+1) + 54 - 3 \leq 3a_{t,\ell} + 54.
      \end{multline*}
      }{
    $$m_{t,>\ell}
      \stackrel{\small{(\ref{eq:1})}}{\leq}
      m_{t-1,>\ell} - 3
      \stackrel{\small{\text{invariant}}}{\leq}
      3 a_{t-1,\ell} + 54 - 3
      \stackrel{\small{\text{Lemma}~\ref{lem:alive}}}{\leq}
      3 (a_{t,\ell}+1) + 54 - 3 \leq 3a_{t,\ell} + 54.$$
    }
  \end{itemize}

  The other case to consider is when the bad level $\ell$ at time $t$
  satisfies $\ell > \ls$. We claim that such a level $\ell$ at time
  $t-1$ must have either been good, or satisfies $m_{t-1, >\ell} < 36$.
  Indeed, by the choice of $\ls$, if $m_{t-1, >\ell} \geq 36$ we must
  have had $m_{t-1,>\ell} \leq 3\nalive(\Cnew_\ell) \leq 3a_{t-1,\ell}$,
  and hence would be good. Hence we just have to consider these two
  cases.
  \begin{itemize}
  \item Suppose $\ell$ was good at time $t-1$, i.e., $m_{t-1, >\ell}
    \leq 3 a_{t-1,\ell}$. Then
   \stocoption{
    \begin{multline*}
    m_{t,>\ell}
    \stackrel{\small{(\ref{eq:2})}}{\leq}
    m_{t-1,>\ell} + 18
    \stackrel{\small{\text{goodness}}}{\leq}
    3 a_{t-1,\ell} + 18 \\
    \stackrel{\small{\text{Lemma}~\ref{lem:alive}}}{\leq}
    3 (a_{t,\ell} + 1) + 18 < 3a_{t,\ell} + 54.
    \end{multline*}
   }{
    $$ m_{t,>\ell}
    \stackrel{\small{(\ref{eq:2})}}{\leq}
    m_{t-1,>\ell} + 18
    \stackrel{\small{\text{goodness}}}{\leq}
    3 a_{t-1,\ell} + 18
    \stackrel{\small{\text{Lemma}~\ref{lem:alive}}}{\leq}
    3 (a_{t,\ell} + 1) + 18 < 3a_{t,\ell} + 54.$$
   }
  \item $m_{t-1, >\ell} < 36$: in this case,
    \[ m_{t,>\ell}
    \stackrel{\small{(\ref{eq:2})}}{\leq}
    m_{t-1,>\ell} + 18
    < 36 + 18  = 54 \leq 3a_{t,\ell} + 54. \]
  \end{itemize}
  This completes the proof of the invariant.
\end{proof}

To recap, the invariant says that for bad levels, the number of edges
added to $F_t'$ above that level is at most thrice the number of active
components plus an additive constant. This is contrast to good levels,
where the additive constant is missing.

\subsubsection{The Final Accounting}

\begin{lemma}[Lipschitz]
  \label{lem:small-change}
  The number of edges in $F_{t-1} \symdif F_t$ is at most $O(1)$.
\end{lemma}

\begin{proof}
  Recall the difference between $F_t$ and $F_t'$ is that the latter
  contains edges added after there is a single alive cluster, and until
  there is a single non-dead cluster. In particular, the difference $|
  F_t' \setminus F_t| = m_{t,>r_t}$. By the invariant, since $a_{t,r_t}
  = 1$, this difference is at most $55$. Moreover,
  \stocoption{
  \begin{multline*}
    | F_{t-1} \symdif F_t | \leq | F'_{t-1} \symdif F'_t | + |F_t'
    \setminus F_t| + |F_{t-1}' \setminus F_{t-1}| \\ \leq | F'_{t-1}
    \symdif F'_t | + 110.
  \end{multline*}
  }{
  \[ | F_{t-1} \symdif F_t | \leq | F'_{t-1} \symdif F'_t | +
  |F_t' \setminus F_t| + |F_{t-1}' \setminus F_{t-1}| \leq | F'_{t-1}
  \symdif F'_t | + 110. \]
  }

  By the refinement property (Lemma~\ref{lem:refinement}) we know that
  $|F_t'| \leq |F_{t-1}'|$. And Lemma~\ref{lem:safe-stable} and
  Fact~\ref{fct:most-safe} show that $|F_{t-1}' \setminus F_{t}'|$ is at
  most $3|\sZ_t| -1 = 17$.  Hence,
  \[ | F'_{t-1} \symdif F'_t | = | F'_{t-1} \setminus F'_t | + | F'_t
  \setminus F'_{t-1} | \leq 2| F'_{t-1} \setminus F'_t | \leq 34. \]
  This proves the Lipschitz property.
\end{proof}

\begin{theorem}
  \label{thm:cost-change}
  For any time $t$, the cost of $F_t$ is at most $O(1)$ times the
  optimal Steiner tree cost on the non-deleted nodes $[t+1,n]$.
\end{theorem}

\begin{proof}
  The proof is very similar to that of Lemma~\ref{lem:cost-amort}. The
  lower bound on the Steiner tree is again at least
  \[ \sum_{\ell} 2^{\ell - 2} \cdot a_{t, \ell} \cdot
  \mathbf{1}(a_{t,\ell} \geq 2) \quad = \quad \sum_{\ell = 1}^{r_t - 1}
  2^{\ell - 2} \cdot a_{t, \ell}. \] The edges of the forest $F_t$
  output by our algorithm are added in levels $\ell \in \{1, \ldots,
  r_t\}$, and have total cost at most $\sum_{\ell = 1}^{r_t} 2^\ell
  \cdot m_{t,>\ell}$. By Lemma~\ref{lem:invariant}, these quantities are
  within a constant factor of each other, which completes the proof.
\end{proof}

\subsection{Elementary Fact}
\label{sec:elementary-facts}

Finally, one elementary fact, capturing that every forest must have a
large number of low-degree vertices.

\begin{fact}
  \label{fact:tree}
  Suppose we are given a forest $F$ with at least $36$ edges on some set
  $V$ of vertices, where $V$ is partitioned into sets $A$ and $B$. If
  the number of edges in $F$ is more than $3|A|$, then there must exist
  a set $S \sse B$ of $6$ nodes where the degrees of nodes in $S$ are
  either one or two.
\end{fact}

\begin{proof}
  Let $V(F)$ denote the nodes in $F$ that have degree at least $1$.
  Consider the set $L \sse V(F)$ of the ``low'' degree nodes, i.e., the
  degree~$1$ or degree-$2$ nodes in $F$. At least half the nodes in
  $V(F)$ must lie in $\ell$. (Indeed, all nodes in $V(F) - L$ contribute
  degree at least $3$, and the nodes in $\ell$ contribute degree at least
  $1$, and the average degree of nodes in $V(F)$ is strictly less than
  $2$ since it is a forest.) So $|L| \geq |V(F)|/2 \geq |E(F)|/2$. Since
  $|A| \leq |E(F)|/3$, we have that $|L \setminus A| \geq |E(F)|/6 \geq
  6$; these are chosen to be in $S$.
% Choose the case with more nodes to be
%   $S$.
% Now find an independent set of size
%   at least half the size (again using the fact that we have a forest,
%   and the nodes have low degree), which gives us a set of size at least
%   $\lceil \frac{25}{2} \rceil> 10$.
\end{proof}

%%% Local Variables:
%%% mode: latex
%%% TeX-master: "deletion"
%%% End:

%\input{deletion-old}

\section{The Fully Dynamic Case}
\label{sec:fully-d}

\newcommand{\add}{\textsf{add}}
\newcommand{\delete}{\textsf{del}}
\newcommand{\del}{\delete}
\newcommand{\That}{\smash{\widehat{T}}}

We now consider the fully-dynamic case, where the input sequence has
both additions and deletions. Hence each request $\sigma_t$ is either
$(\add, t)$, or $(\delete, t')$ for some $t' < t$.  We assume that each
vertex that is added is a ``new'' vertex, and hence has a new
index. Moreover, this means there is no point to deleting vertices multiple
times, each vertex can be assumed to be deleted at most once. 

In the fully-dynamic case, observe that the process can go on
indefinitely and the metric can be arbitrarily large, so $n$ will denote
some arbitrary instance in time, instead of denoting the size of the
metric as in the previous section.  Let $V_n = \{ t \in [n] \mid
\sigma_t = (\add, t) \}$ be the set of vertices that have appeared until
time $n$. Let $D_n = \{ s \in [n] \mid \exists t \in [n] ~s.t.~ \sigma_t
= (\delete, s) \}$ be the vertices that have been deleted until time $n$;
since each vertex is deleted at most once, this is well-defined.  Let $A_n
= V_n \setminus D_n$ be the ``alive'' vertices at time $n$.

\interfootnotelinepenalty=10000

We will assume that the inter-point distances are specified in the
following particular manner --- this will be convenient for us in the
following analysis. (In Section~\ref{sec:dist-spec}, we argue this does
not change the problem by more than a constant factor.) Let $d_t(\cdot,
\cdot)$ be the distances between the vertices in $V_t$.  If we see
$(\add, t+1)$, we are given the distances from $t+1$ to all vertices in
$A_t \sse V_t$: i.e., to only the alive vertices. The guarantee we have
is that the newly given distances form a metric along with the old
distances, and hence do not violate the triangle inequality. The
distances from $t+1$ to vertices in $D_t$ must be inferred using the
triangle inequality: $d_{t+1}(t+1, s) = \min_{s' \in A_t} (d(t+1,s') +
d_t(s',s))$. Note that the former summand is a new distance given as
input, the second summand is inductively defined.

Let $T_t$ be the tree at time $t$, and $V(T_t)$ denote the set of
vertices in it.  The following lemma is immediate:
\begin{lemma}
  \label{lem:metric}
  The distances $d_t$ satisfy the following properties:
  \begin{OneLiners}
  \item[(a)] The closest distance from $t+1$ to vertices in $V(T_t)$ is to
    some alive vertex; i.e. some vertex in $A_t$.
  \item[(b)] The metric $d_{t+1}$ restricted to $V_t$ is the same as $d_t$.
  \end{OneLiners}
\end{lemma}

The tree $T_t$ at some time $t$ is valid if it uses any vertices in $V_t$,
whether they are alive or dead, but it contains all the alive vertices
$A_t$. (Hence $T_t$ is a Steiner tree on $A_t$, with $D_t$ being the
Steiner vertices.)
Consider the two
cases for request $\sigma_{t+1}$:
\begin{itemize}
\item Case I: $\sigma_{t+1} = (\add, t+1)$. In this case we are now
  given the distances from $t+1$ to all vertices in $A_t$, and hence can
  infer the new distance metric $d_{t+1}(\cdot, \cdot)$. We now must add
  at least one edge from $t+1$ to $V(T_t)$ to get connectivity, and then
  are allowed to make any edge swaps, and also potentially drop some
  deleted vertices from the tree to get tree $T_{t+1}$.
\item Case II: $\sigma_{t+1} = (\delete, s)$. We mark the vertex $s \in
  V(T_t)$ as deleted. We are allowed to make any edge swaps, and also
  potentially drop some deleted vertices from the tree to get tree $T_{t+1}$.
\end{itemize}
Finally, the cost of tree $T_t$ is $\cost(T_t) := \sum_{e \in T_t}
d_t(e)$, the sum of lengths of the edges in $T_t$. We call this tree
$\alpha$-competitive if $\cost(T_t) \leq \alpha\cdot \OPT(A_t)$, i.e.,
it costs not much more than the minimum cost Steiner tree on the alive
vertices. The algorithm is said to be $\alpha$-competitive in the fully-dynamic
model if it maintains a tree that is $\alpha$-competitive at all times,
when the input consists of both additions and deletions.

The main theorem of this section is the following:
\begin{theorem}
  \label{thm:fully-d}
  There is a $4$-competitive algorithm for Steiner tree in the
  fully-dynamic model that, for every $t$, performs at most $O(t)$ edge
  additions and deletions in the first $t$ steps.
\end{theorem}

Before we give the algorithm, let us define $c$-swaps and
$c$-stability. For some Steiner tree $T$ on the terminals in $A_t$,
suppose there exist $e \in E(T)$ and $f \not\in E(T)$ such that (i)~the
graph $T - e + f$ is also a Steiner tree on $A_t$, and (ii)~$d_t(e) \geq
c\cdot d_t(f)$. Then we say that $(e,f)$ is a \emph{valid $c$-swap}, and
performing the valid $(e,f)$ swap means changing the current tree from
$T$ to $T - e +f$. A tree is $c$-stable if there do not exist any valid
$c$-swaps.

Following Imase and Waxman~\cite{IW91}, a tree $T$ with vertex set $V(T)$ is called an
\emph{extension tree} for a set of vertices $S$ if (i) it is a Steiner tree on
 $S$---i.e., $S \sse V(T)$, and (b) all  Steiner vertices in $T$, i.e., vertices in $V(T)\S$
are of degree strictly greater than $2$. Given a Steiner tree $T$ that
is not an extension tree (i.e., $T$ has Steiner vertices of degree $1$ or
$2$), the following operations produce an extension tree $T'$.  For any
degree-$1$ Steiner vertex (i.e., leaf Steiner vertex), delete the vertex and
its incident edge. For any degree-$2$ Steiner vertex $u$ with edges to
$v,w$, delete the vertex $u$ and edges $(u,v), (u,w)$, and add the edge
$(v,w)$. Note that such an operation might create more low-degree vertices:
repeat the process on these vertices until the resulting tree is an
extension tree for $S$.

Our algorithm is the following: 
\begin{itemize}
\item For an addition $\sigma_t = (\add, t)$, attach $t$ to the closest
  vertex $p_t$ from $V(T_{t-1})$. Call the edge $(t,p_t)$ the greedy edge
  for time $t$. By \lref[Lemma]{lem:metric}, the vertex $p_t$ is alive. Now
  perform any valid $2$-swaps until we get a $2$-stable tree.
\item For a deletion $\sigma_t = (\delete, s)$, mark $s$ as a Steiner
  vertex in $T_{t-1}$. Convert this Steiner tree on $A_t = A_{t-1}
  \setminus \{s\}$ to an extension tree as described above. Perform any
  valid $2$-swaps until the tree is $2$-stable. This might create
  low-degree vertices, so repeat these two steps iteratively until we get a $2$-stable
  extension tree on the vertex set~$A_t$. Note that this process will terminate because 
  during edge swaps, we are reducing the cost of the tree, and during conversion to an 
extension tree, we are removing some vertices which will not appear again. 
\end{itemize}

Recall that for a set of vertices $S$, $\OPT(S)$ denotes the cost of the optimal Steiner tree on $S$. Let 
 $\MST(S)$ denote the cost of the minimum spanning tree on $S$. Let 
$\cost(\MST(S))$ denote the cost of this tree. 
The argument about the cost of the tree follows from known results~\cite[Lemma~5]{IW91}:
\begin{theorem}
  \label{thm:iwlem5}
  If $T = (V,E)$ is a $c$-stable extension tree for a set of vertices $S$, then 
  \[ \cost(T) \leq 2c\cdot \cost(\MST(S)) \leq 4c\cdot \OPT(S). \] 
\end{theorem}
This shows that the tree maintained by the algorithm is
$4$-competitive. To prove \lref[Theorem]{thm:fully-d}, it suffices to now
bound the number of edge additions and deletions performed during the algorithm. The following 
lemma follows from the fact that the closest vertex to a newly arriving vertex is one of the 
alive vertices at that time. 

\begin{lemma}
  \label{lem:greedy-edges}
  For any $n$, consider the algorithm after the first $n$ requests. The
  greedy edges added by the algorithm are the same edges that would be
  added by a greedy algorithm running on the subsequence of just the
  additions (and none of the deletions) in the sequence
  $\sigma_{1\cdots n}$.
\end{lemma}

\begin{corollary}
  \label{clm:g-product}
  For any $n$, let $E_g$ be the set of greedy edges added by the
  algorithm on input sequence $\sigma_{1\cdots n}$. Then
  \[ \prod_{(t,p_t) \in E_g} d_n(t, p_t) = \prod_{(t,p_t) \in E_G}
  d_t(t, p_t) \leq 4^{|V_n|} \cdot \prod_{e \in \MST(V_n)} d_n(e). \]
\end{corollary}

\begin{proof}
 This
 follows from~\cite[Theorem~5.1]{GuGK13}, which bounds the product
  of the greedy edge lengths added in the sequence $\sigma_{1\cdots n}$ 
  to the edge lengths of the minimum spanning
  tree of the added vertices $V_n$ -- this result assumes that we are inserting vertices
  only. But the lemma above shows that we can indeed make such an assumption without 
affecting the set of greedy edges which get added. 
\end{proof}

It is now convenient to define a slightly different process, where we
maintain a spanning tree $\That_t$ on all vertices in $V_t$, instead of a
Steiner vertex on terminals in $A_t$. Some of the edges in $\That_t$ will
be colored red, and others black. One invariant will be that deleting
all the red edges in $\That_t$ will leave exactly the tree $T_t$. Hence
the red edges give us a forest, where each tree in this forest 
contains a single vertex from~$T_t$.

Suppose we have inductively defined $\That_t$ thus far. There are four
different operations:
\begin{OneLiners}
\item[(a)] Any greedy edges added to $T_t$ are also added to $\That$ and
  colored black.
\item[(b)] Any swaps done in $T_t$ are also mimicked in
  $\That_t$---observe that these are swaps between black edges.
\item[(c)] If we delete some degree-$1$ vertex from $T_t$, we merely mark
  this edge as red in $\That_t$.
\item[(d)] If we delete a degree-$2$ vertex $u$ in $T_t$, and connect its
  neighbors $(v,w)$ by an edge, in $\That_t$ we also add the new black
  edge $(v,w)$, delete the black edges $(u,v), (u,w)$, and add a red
  edge from $u$ to the closer of $\{v,w\}$.
\end{OneLiners}

Note that all these moves maintain that $\That_t$ is a spanning tree on
the vertices in $V_t$; all edges of $T_t$ are contained in it, and are
colored black. Define the potential of any forest $F$ on the metric
$d_n$ as
\begin{gather}
  \Phi_n(F) := \prod_{e \in E(F)} d_n(e).
\end{gather}
Hence, \lref[Corollary]{clm:g-product} says that $\Phi_n(E_g) \leq
4^{|V_n|} \cdot \Phi_n(\MST(V_n))$. (Notice that the greedy edges $E_g$ do form
a forest--in fact, a spanning tree--on $V_n$.) Let us track how the
potential of the tree $\That_t$ changes.

\begin{lemma}
  \label{lem:pot-count}
  \[ \Phi_n(\That_n) \leq \Phi_n(E_g) \cdot \left( \frac12 \right)^{n_b}
  \cdot 2^{n_d} \] where $n_b$ is the number of $2$-swaps (i.e., number
  of invocations of operation~(b)), and $n_d$ is the number of
  invocations of operation~(d).
\end{lemma}

\begin{proof}
  The change in the product due to operation~(a) is captured by the
  product of the greedy edges. Operation~(b) causes some edge to be
  replaced by an edge of at most half the length, which accounts for
  $(1/2)^{n_b}$. Operation~(c) does not change the product, only the
  color of an edge. Operation~(d) essentially replaces the longer of
  edges $(u,v), (u,w)$---say the longer one is $(u,w)$---by the edge
  $(v,w)$. By the triangle inequality, $(u,w)$ has length $d_n(v,w) \leq
  d_n(u,v) + d_n(u,w) \leq 2\,d(u,w)$. Hence the product of edge lengths
  increases by at most a factor of $2$. This accounts for $2^{n_d}$.
\end{proof}

Putting \lref[Corollary]{clm:g-product} and \lref[Lemma]{lem:pot-count}
together, and using that $\That_n$ is a spanning tree on $V_n$, we get
\begin{gather}
\label{eqn:key1}
  \frac{\Phi_n(\That_n)}{\Phi_n(\MST(V_n))} \leq 4^{|V_n|} \cdot
  2^{-n_b + n_d}.
\end{gather}

As observed in ~\cite[Lemma~5.2]{GuGK13}, the fact that $\That_n$  is a spanning 
tree on $V_n$  and $\MST(V_n)$ is a minimum spanning tree on $V_n$ implies that 
$\Phi_n(\That_n) \geq \Phi_n(\MST(V_n)).$ Combining this with~(\ref{eqn:key1}), we get
\[ n_b \leq 2\,|V_n| + n_d.  \]

Note that the algorithm performs at most $|V_n|$ edge deletions, since
each execution of operations~(c) and~(d) causes one edge deletion. Also,
each operation~(d) also causes one edge swap (in addition to the edge
deletion), as does an execution of operation~(b). Hence the total number
of swaps is at most
\[ n_b + n_d \leq 2\,|V_n| + 2\,n_d. \] Finally, $n_d = n - |V_n|$,
because there are a total of $n$ requests and $|V_n|$ of them are vertex
additions (so the rest must be deletions). This means the total number
of edge swaps in the first $n$ requests is $2n$, which completes the
proof of \lref[Theorem]{thm:fully-d}.

\subsection{A Word about the Distance Specification}
\label{sec:dist-spec}

Recall that when a new point was added, we specified the distances from
this new point to the old points in a particular fashion. Let us recall
this again. Suppose $d_t(\cdot, \cdot)$ are the current distances
between the vertices in $V_t$.  If we see $(\add, t+1)$, we are given
the distances from $t+1$ to all vertices in $A_t \sse V_t$: i.e., to
\emph{only the alive vertices}. The guarantee we have is that the newly
given distances form a metric along with the old distances, and hence do
not violate the triangle inequality. The distances from $t+1$ to
vertices in $D_t$ must be inferred using the triangle inequality:
$d_{t+1}(t+1, s) = \min_{s' \in A_t} (d(t+1,s') + d_t(s',s))$. Note that
the former summand is a new distance given as input, the second summand
is inductively defined.

Perhaps a more natural model is where we are given distances to all the
previous vertices (both alive and deleted), again subject to the
triangle inequality. We now claim the two models are the same up to
constant factors, and hence it is fine to work with the former
model. Indeed, suppose when we see $(\add, t+1)$, we are told distances
$d'(t+1, x)$ for all $x \in V_t$, and this gives us a metric
$d'_{t+1}(\cdot, \cdot)$ on $V_t$. We could then ignore the distances to
the already-deleted vertices, define $d(t+1, y) := d'(t+1,y)$ for all
alive vertices $y \in A_t$ and extend it by the triangle inequality as
above to get the distances $d_{t+1}(\cdot, \cdot)$ on all of
$V_{t+1}$. Clearly the distances $d_{t+1} \geq d_{t+1}'$, so the cost of
our tree according to the distances $d_{t+1}$ is at least the cost
according to the actual distances $d'_{t+1}$. Moreover, this definition
inductively maintains $d_{t+1}(x,y) = d'_{t+1}(x,y)$ for all $x,y \in
A_{t+1}$, so the cost of the optimal Steiner tree on $A_{t+1}$ using the
actual metric $d'_{t+1}$ is at least half the cost of the MST on
$A_{t+1}$ with respect to $d'_{t+1}$ (and hence also with respect to
$d_{t+1}$). This completes the proof that working in the
distances-specified-to-alive-points only changes the competitive ratio
by a factor of $2$.

%%% Local Variables: 
%%% mode: latex
%%% TeX-master: "deletion"
%%% End: 

\section{Discussion}

Several interesting questions remain unanswered. We do not know how to
get a non-amortized constant competitive algorithm for the fully-dynamic
case which makes $O(1)$ swaps per insertion or deletion. Obtaining
similar results for the Steiner forest problem (even in the amortized
setting, even for insertions only) remains an interesting open problem.

\ifstoc
 \bibliographystyle{abbrv}
 {\bibliography{bibonline}}
\else
 \bibliographystyle{alpha}
 {\small \bibliography{bibonline}}

\newcommand{\etalchar}[1]{$^{#1}$}
\begin{thebibliography}{MSVW12}

\bibitem[AAB04]{AAB96}
Baruch Awerbuch, Yossi Azar, and Yair Bartal.
\newblock On-line generalized {S}teiner problem.
\newblock {\em Theoret. Comput. Sci.}, 324(2-3):313--324, 2004.

\bibitem[AKR95]{AKR95}
Ajit Agrawal, Philip Klein, and R.~Ravi.
\newblock When trees collide: an approximation algorithm for the generalized
  {Steiner} problem on networks.
\newblock {\em SIAM J. Comput.}, 24(3):440--456, 1995.

\bibitem[BC97]{BC97}
Piotr Berman and Chris Coulston.
\newblock On-line algorithms for {Steiner} tree problems.
\newblock In {\em STOC}, pages 344--353, 1997.

\bibitem[BN07]{BN-mono}
Niv Buchbinder and Joseph Naor.
\newblock The design of competitive online algorithms via a primal-dual
  approach.
\newblock {\em Found. Trends Theor. Comput. Sci.}, 3(2-3):front matter, 93--263
  (2009), 2007.

\bibitem[GGK13]{GuGK13}
Albert Gu, Anupam Gupta, and Amit Kumar.
\newblock The power of deferral: maintaining a constant-competitive steiner
  tree online.
\newblock In {\em STOC}, pages 525--534, 2013.

\bibitem[GW95]{GW95}
Michel~X. Goemans and David~P. Williamson.
\newblock A general approximation technique for constrained forest problems.
\newblock {\em SIAM J. Comput.}, 24(2):296--317, 1995.

\bibitem[IW91]{IW91}
Makoto Imase and Bernard~M. Waxman.
\newblock Dynamic {Steiner} tree problem.
\newblock {\em SIAM J. Discrete Math.}, 4(3):369--384, 1991.

\bibitem[KLS05]{KLS05}
Jochen K{\"o}nemann, Stefano Leonardi, and Guido Sch{\"a}fer.
\newblock A group-strategyproof mechanism for steiner forests.
\newblock In {\em SODA}, pages 612--619, 2005.

\bibitem[KLSvZ08]{KLSZ08}
Jochen K{\"o}nemann, Stefano Leonardi, Guido Sch{\"a}fer, and Stefan H.~M. van
  Zwam.
\newblock A group-strategyproof cost sharing mechanism for the steiner forest
  game.
\newblock {\em SIAM J. Comput.}, 37(5):1319--1341, 2008.

\bibitem[{\L}OP{\etalchar{+}}13]{LOPSZ13}
Jakub {{\L}acki}, Jakub O\'cwieja, Marcin Pilipczuk, Piotr Sankowski, and Anna
  Zych.
\newblock Dynamic steiner tree in planar graphs.
\newblock {\em CoRR}, abs/1308.3336, abs/1308.3336, 2013.

\bibitem[MSVW12]{MSVW12}
Nicole Megow, Martin Skutella, Jos{\'e} Verschae, and Andreas Wiese.
\newblock The power of recourse for online {MST} and {TSP}.
\newblock In {\em ICALP (1)}, pages 689--700, 2012.

\end{thebibliography}
\fi

\appendix

%\stocoption{
%\input{app-proofs}
%}{}

\end{document}